\documentclass[letterpaper, 10 pt, conference]{ieeeconf}

\IEEEoverridecommandlockouts                              

\usepackage{amsmath,amssymb,amsfonts}

\usepackage{amsthm}
\usepackage{algorithm}
\usepackage{algpseudocode}

\usepackage{graphicx}
\usepackage{textcomp}
\usepackage{xcolor}
\usepackage{cite}
\usepackage{booktabs}
\usepackage{multirow}
\usepackage{url}
\usepackage{bm}
\usepackage{subcaption}
\usepackage{balance}

\newtheoremstyle{cdcplain}{}{}{}{}{\itshape}{.}{ }{}
\newtheoremstyle{cdcdefn}{}{}{}{}{\itshape}{.}{ }{}

\theoremstyle{cdcplain}
\newtheorem{theorem}{Theorem}

\newtheorem{proposition}{Proposition}

\theoremstyle{cdcdefn}
\newtheorem{definition}{Definition}
\newtheorem{assumption}{Assumption}
\newtheorem{problem}{Problem}
\newtheorem{remark}{Remark}

\newcommand{\NUM}[1]{#1}

\title{\LARGE \bf Conformalized Data-Driven Reachability Analysis\\with PAC Guarantees}

\author{Yanliang Huang$^{*,1}$, Zhen Zhang$^{*,1}$, Peng Xie$^{1}$, Zhuoqi Zeng$^{2}$, Amr Alanwar$^{1}$%
\thanks{$^{*}$Equal contribution.}%
\thanks{$^{1}$Technical University of Munich, Munich, Germany.
        \texttt{\{yanliang.huang, zhenzhang.zhang, p.xie, alanwar\}@tum.de}}%
\thanks{$^{2}$Hainan Bielefeld University of Applied Sciences, Hainan, China.
        \texttt{zhuoqi.zeng@hainan-biuh.edu.cn}}%
}

\begin{document}

\maketitle
\thispagestyle{empty}
\pagestyle{empty}

\begin{abstract}
Data-driven reachability analysis computes reachable sets directly from noisy data.
Existing matrix-zonotope reachability methods require a known noise bound, and their nonlinear extensions require a Lipschitz constant.
We propose Conformalized Data-Driven Reachability, which calibrates the open-loop prediction error of a least-squares model with training-conditional split conformal prediction and propagates the initial set and the input set with zonotopes.
The resulting sets contain the whole trajectory with a probably approximately correct (PAC) guarantee that requires only independent and identically distributed trajectories, and the number of calibration trajectories does not depend on the horizon.
The prediction error depends on the initial states and the controller only through the model error.
The method extends to measurement noise with a known quantile and to nonlinear systems through local affine models.
In experiments, the sets keep an average coverage above 99~\% when the initial states and the controller change, whereas sets calibrated directly on the states lose coverage.
\end{abstract}

\section{Introduction}\label{sec:introduction}

Reachability analysis computes the set of all states that a dynamical system can reach from a given initial set under admissible inputs, and is a standard tool for verifying safety in cyber-physical systems \cite{althoff2021set}.
Model-based approaches require an explicit system model, which is often unavailable or expensive to obtain.
Data-driven reachability analysis \cite{alanwar2023datadriven} constructs over-approximating reachable sets directly from measured trajectories without a given system model.

The methods in \cite{alanwar2023datadriven} are matrix-zonotope reachability methods, which enclose all models that are consistent with the data and a given bound on the process noise, so they require this bound as input.
In practice such a bound is chosen from data or from engineering judgement, and the resulting guarantee holds only if the chosen bound covers every noise realization.
When the noise has unbounded support, no finite bound exists.
For nonlinear systems, the extension of these methods in \cite{alanwar2023datadriven} requires a Lipschitz constant of the dynamics and the covering radius of the data.

Probabilistic approaches avoid such prior bounds.
Scenario optimization \cite{devonport2020scenario}, holdout estimates \cite{dietrich2025holdout}, Christoffel functions \cite{devonport2021christoffel,tebjou2023christoffel} and conformal prediction \cite{hashemi2023datadriven} construct sets that contain the states with high probability.
These guarantees refer to the distribution under which the data were collected.
When the sets are used to verify a controller, the inputs are generated by that controller and the initial states by the task.
The distribution of the inputs and the initial states then differs from the data distribution whenever the data were collected with other inputs, for instance with exploratory inputs recorded before the controller was designed.
Controllers that optimize performance often operate at the limits of the input set, for instance through saturation, which drives the state towards the boundary of the reachable set.

In this paper, we propose Conformalized Data-Driven Reachability (CDDR), which calibrates the open-loop prediction error of a least-squares model with training-conditional split conformal prediction \cite{vovk2012conditional} and propagates the initial set and the input set with zonotopes.
The prediction error depends on the initial states and the controller only through the model error, whereas direct calibration of the reachable sets does not have this property.

\subsection{Related Work}\label{sec:related}

Conformal prediction has been applied to planning and reachability.
Lindemann et al.~\cite{lindemann2023safe} use conformal prediction for safe planning among dynamic agents, Hashemi et al.~\cite{hashemi2023datadriven} compute reachable sets of stochastic systems from learned surrogate models, and Stankeviciute et al.~\cite{stankeviciute2021conformal} and Cleaveland et al.~\cite{cleaveland2024conformal} construct prediction regions for multi-step forecasts.
These works provide marginal coverage, which holds on average over calibration draws.
Training-conditional guarantees \cite{vovk2012conditional,park2020pac} hold for the realized set with a prescribed confidence and are used in this paper.
Weighted conformal prediction \cite{tibshirani2019conformal} corrects for distribution shift with known likelihood ratios, whereas the score of CDDR depends on the new initial distribution and controller only through the model error.

Devonport and Arcak~\cite{devonport2020scenario} estimate reachable sets by scenario optimization, Devonport et al.~\cite{devonport2021christoffel} use Christoffel functions with probably approximately correct (PAC) guarantees, and Tebjou et al.~\cite{tebjou2023christoffel} combine Christoffel functions with conformal prediction.
Dietrich et al.~\cite{dietrich2025holdout} certify reachable set estimates with the holdout method, and Dietrich et al.~\cite{dietrich2026pac} relate the PAC bounds of conformal prediction, scenario optimization and the holdout method.
Stochastic model predictive control with probabilistic reachable sets \cite{hewing2018stochastic,hewing2020scenario} also splits the state into a nominal part and an error driven only by the noise.
These methods compute the error sets from a known model and the noise distribution or noise samples, whereas CDDR estimates the model and calibrates the error on trajectory data.
Reachset-conformant identification \cite{lutzow2024scalable,lutzow2026reachset} chooses the smallest uncertainty sets of a model whose reachable outputs contain all measured outputs, so the containment is established on the data used for identification.

\subsection{Contributions}\label{sec:contributions}
The contributions of this paper are as follows.
\begin{itemize}
\item We propose a data-driven reachability method for linear systems that requires no noise bound.
A trajectory-level score yields an $(\alpha,\delta)$ PAC guarantee for the whole horizon, and the required number of calibration trajectories does not depend on the horizon.
The score depends on the controller only through the model error.
\item We extend the method to measurement noise with a known quantile and to nonlinear systems through local affine models without a Lipschitz constant.
\item We validate the method on a five-dimensional linear system and a two-dimensional nonlinear system under changes of the initial states and the controller, and compare it with direct calibration of the reachable sets and with matrix-zonotope reachability using the true noise bound.
\end{itemize}

\section{Preliminaries and Problem Formulation}\label{sec:problem}

\subsection{Notation}\label{sec:sets}

The sets of real and natural numbers are denoted by $\mathbb{R}$ and $\mathbb{N}$.
The transpose and the Moore-Penrose pseudoinverse of a matrix $X$ are denoted by $X^\top$ and $X^\dagger$.
The identity matrix of dimension $d$ is denoted by $I_d$, and $\mathrm{diag}(a_1, \dots, a_d)$ is the diagonal matrix with diagonal entries $a_1, \dots, a_d$.
For a list $a$, the element $i$ of $a$ is denoted by $a^{(i)}$, and for a vector $z$, the component $l$ of $z$ is denoted by $z_l$.
For $n \in \mathbb{N}$ and $p \in [0,1]$,
$B_{n,p}(k) = \sum_{i=0}^{k} \binom{n}{i} p^i (1-p)^{n-i}$
denotes the distribution function of the Binomial distribution with $n$ trials and success probability $p$.

\subsection{Preliminaries}\label{sec:prelim}\label{sec:tc}

Zonotopes \cite{kuhn1998rigorously} are used as the set representation throughout this paper.

\begin{definition}[Zonotope~\cite{kuhn1998rigorously}]\label{def:zono}
Given a dimension $d \in \mathbb{N}$, a center $c_\mathcal{Z} \in \mathbb{R}^{d}$ and $\gamma_\mathcal{Z} \in \mathbb{N}$ generator vectors in a generator matrix $G_\mathcal{Z} = \big[g_\mathcal{Z}^{(1)} \ \dots \ g_\mathcal{Z}^{(\gamma_\mathcal{Z})}\big] \in \mathbb{R}^{d \times \gamma_\mathcal{Z}}$, a zonotope is defined as
\begin{equation}\label{eq:zono}
\mathcal{Z} = \Big\{ x \in \mathbb{R}^{d} \ \Big|\ x = c_\mathcal{Z} + \sum_{i=1}^{\gamma_\mathcal{Z}} \beta^{(i)} g_\mathcal{Z}^{(i)},\ -1 \le \beta^{(i)} \le 1 \Big\} .
\end{equation}
We use the shorthand notation $\mathcal{Z} = \langle c_\mathcal{Z}, G_\mathcal{Z} \rangle$.
\end{definition}

Let $L \in \mathbb{R}^{m \times d}$ be a linear map. Then $L\mathcal{Z} = \langle L c_\mathcal{Z}, L G_\mathcal{Z} \rangle$.
Given two zonotopes $\mathcal{Z}_1 = \langle c_{\mathcal{Z}_1}, G_{\mathcal{Z}_1} \rangle$ and $\mathcal{Z}_2 = \langle c_{\mathcal{Z}_2}, G_{\mathcal{Z}_2} \rangle$, we use $+$ to denote the Minkowski sum, which is computed as
\begin{equation}\label{eq:minkowski}
\begin{aligned}
\mathcal{Z}_1 + \mathcal{Z}_2 &= \{ z_1 + z_2 \mid z_1 \in \mathcal{Z}_1,\ z_2 \in \mathcal{Z}_2 \} \\
&= \big\langle c_{\mathcal{Z}_1} + c_{\mathcal{Z}_2}, [G_{\mathcal{Z}_1}, G_{\mathcal{Z}_2}] \big\rangle .
\end{aligned}
\end{equation}


Split conformal prediction constructs prediction sets from a model fitted on training data and a threshold computed on separate calibration data \cite{vovk2005algorithmic,lei2018distribution}.
For each calibration sample, a score measures the deviation between the sample and the prediction of the model, and large scores indicate poorly predicted samples.
An order statistic of the calibration scores serves as a threshold, and the prediction set contains all points whose score does not exceed the threshold.
If the scores $s^{(1)}, \dots, s^{(n)}$ of the calibration samples and the score $s^{(n+1)}$ of a new sample are exchangeable, which means that their joint distribution does not change under permutations, the $\lceil (n+1)(1-\alpha) \rceil$-th smallest calibration score $\hat q_{\mathrm{m}}$ satisfies
\begin{equation}\label{eq:marginal}
\mathbb{P}\big[ s^{(n+1)} \le \hat q_{\mathrm{m}} \big] \ge 1 - \alpha .
\end{equation}
The probability in \eqref{eq:marginal} is taken over the calibration samples and the new sample together.
For a fixed calibration set, the coverage $\mathbb{P}\big[ s^{(n+1)} \le \hat q_{\mathrm{m}} \mid s^{(1)}, \dots, s^{(n)} \big]$ is a random variable, and for continuous scores it falls below $1-\alpha$ with probability $43.6$~\% for $n = 100$ and $\alpha = 0.05$.
A reachable set is computed once from the available data and therefore requires a guarantee for the realized threshold.
The training-conditional version of split conformal prediction provides such a guarantee by selecting a higher order statistic.
The following result restates \cite[Prop.~2b]{vovk2012conditional}.

\begin{proposition}\label{prop:tc}
Let $s^{(1)}, \dots, s^{(n+1)}$ be independent and identically distributed real random variables and let $\alpha, \delta \in (0,1)$ satisfy $B_{n,\alpha}(0) = (1-\alpha)^n \le \delta$.
Define
\begin{equation}\label{eq:kstar}
k^\ast = \max\big\{ k \in \{0, \dots, n-1\} \ \big|\ B_{n,\alpha}(k) \le \delta \big\}
\end{equation}
and let $\hat q$ be the $(k^\ast+1)$-th largest element of $\{s^{(1)}, \dots, s^{(n)}\}$.
Then
\begin{equation}\label{eq:tc}
\mathbb{P}\Big[\, \mathbb{P}\big[ s^{(n+1)} > \hat q \ \big|\ s^{(1)}, \dots, s^{(n)} \big] > \alpha \,\Big] \le B_{n,\alpha}(k^\ast) \le \delta .
\end{equation}
The condition $B_{n,\alpha}(0) \le \delta$ holds if and only if $n \ge n_{\min} = \lceil \ln\delta / \ln(1-\alpha) \rceil$.
\end{proposition}

We consider the discrete-time system
\begin{equation}\label{eq:sys_gen}
\begin{aligned}
x(k+1) &= f\big(x(k), u(k)\big) + w(k), \\
y(k) &= x(k) + v(k),
\end{aligned}
\end{equation}
with state $x(k) \in \mathbb{R}^{n_x}$, input $u(k) \in \mathbb{R}^{n_u}$, measured output $y(k) \in \mathbb{R}^{n_x}$, process noise $w(k) \in \mathbb{R}^{n_x}$, measurement noise $v(k) \in \mathbb{R}^{n_x}$ and an unknown function $f : \mathbb{R}^{n_x} \times \mathbb{R}^{n_u} \to \mathbb{R}^{n_x}$.
For linear time-invariant systems without measurement noise, \eqref{eq:sys_gen} reduces to
\begin{equation}\label{eq:sys}
x(k+1) = A_{\mathrm{tr}} x(k) + B_{\mathrm{tr}} u(k) + w(k),
\end{equation}
where $\begin{bmatrix} A_{\mathrm{tr}} & B_{\mathrm{tr}} \end{bmatrix}$ denotes the unknown true system model.

A trajectory of length $N$ is $\tau = \big(x(0), u(0), x(1), \dots, u(N-1), x(N)\big)$.
The initial state $x(0)$ is drawn from an initial distribution $P_0$, and each input $u(k)$ is chosen by a causal policy $\pi$, which may be randomized and may depend on $x(0), \dots, x(k)$.
The process noise $w(k)$ is independent and identically distributed with an unknown distribution $P_w$.
We denote the resulting distribution of trajectories by $\mathcal{P}(P_0, \pi, P_w)$.

\begin{assumption}\label{ass:support}
The initial set $\mathcal{X}_0 \subset \mathbb{R}^{n_x}$ and the input sets $\mathcal{U}_k \subset \mathbb{R}^{n_u}$, $k = 0, \dots, N-1$, are zonotopes.
The support of $P_0$ is contained in $\mathcal{X}_0$, and $\pi$ takes values in $\mathcal{U}_k$ at time $k$.
The noise $w(k)$ is independent of $x(0)$, of the randomness of $\pi$ and of $w(0), \dots, w(k-1)$.
\end{assumption}

No bound on $w(k)$ is assumed, and $P_w$ may have unbounded support.

Calibration data are a held-out set of trajectories that is used only to choose the threshold of the sets.
The calibration trajectories are drawn independently from $\mathcal{P} = \mathcal{P}(P_0, \pi, P_w)$, which we call the data distribution, and the training trajectories are drawn independently of them, for instance also from $\mathcal{P}$.
The training data consist of $K$ input-state trajectories of lengths $T_i$, denoted by $\{u^{(i)}(k)\}_{k=0}^{T_i-1}$ and $\{x^{(i)}(k)\}_{k=0}^{T_i}$, $i = 1, \dots, K$.
For $K = 1$, the training data are collected in the matrices
\begin{equation}\label{eq:data}
\begin{aligned}
X_+ &= [\, x^{(1)}(1) \ \cdots \ x^{(1)}(T_1) \,], \\
X_- &= [\, x^{(1)}(0) \ \cdots \ x^{(1)}(T_1{-}1) \,], \\
U_- &= [\, u^{(1)}(0) \ \cdots \ u^{(1)}(T_1{-}1) \,],
\end{aligned}
\end{equation}
and for $K > 1$ the matrices of the individual trajectories are concatenated horizontally.
We denote the training trajectories by $\mathcal{D}$.
The calibration data $\mathcal{D}_{\mathrm{cal}} = \{\tau^{(1)}, \dots, \tau^{(n)}\}$ consist of $n$ further trajectories of length $N$.
When only the outputs $y(k)$ are measured, the states are replaced by the measured outputs in $\mathcal{D}$ and $\mathcal{D}_{\mathrm{cal}}$.

\subsection{Problem Formulation}\label{sec:formulation}

For a bounded noise zonotope $\mathcal{Z}_w$, the exact reachable set $\mathcal{R}_N$ \cite{alanwar2023datadriven} collects all states $x(N)$ that are reachable from $\mathcal{X}_0$ under $u(k) \in \mathcal{U}_k$ and $w(k) \in \mathcal{Z}_w$.
If the support of $P_w$ is $\mathbb{R}^{n_x}$, as for Gaussian noise, no bounded $\mathcal{Z}_w$ contains all noise realizations, and the states reachable under all noise realizations form the whole state space for every $N \ge 1$.
We therefore consider sets that contain the whole trajectory with high probability.

\begin{definition}[Probabilistic Reachable Sequence]\label{def:prs}
Let $\alpha, \delta \in (0,1)$.
Sets $\mathcal{R}^{\mathrm{CP}}_1, \dots, \mathcal{R}^{\mathrm{CP}}_N$ computed from $\mathcal{D}$ and $\mathcal{D}_{\mathrm{cal}}$ form an $(\alpha, \delta)$ probabilistic reachable sequence for a trajectory distribution $\mathcal{P}'$ if, conditionally on $\mathcal{D}$,
\begin{equation}\label{eq:pac}
\mathbb{P}_{\mathcal{D}_{\mathrm{cal}}}\Big[ \mathbb{P}_{\tau \sim \mathcal{P}'}\big[ x(k) \in \mathcal{R}^{\mathrm{CP}}_k,\ k = 1, \dots, N \big] \ge 1 - \alpha \Big] \ge 1 - \delta ,
\end{equation}
where $\tau$ is independent of $\mathcal{D}$ and $\mathcal{D}_{\mathrm{cal}}$.
\end{definition}

The inner probability in \eqref{eq:pac} is the coverage of the computed sets for a new trajectory, and the outer probability accounts for the randomness of the calibration data.
The trivial choice $\mathcal{R}^{\mathrm{CP}}_k = \mathbb{R}^{n_x}$ satisfies \eqref{eq:pac}, so the sets should also be small.

\begin{problem}\label{prob:main}
Given $\alpha$, $\delta$, $\mathcal{D}$, $\mathcal{D}_{\mathrm{cal}}$, $\mathcal{X}_0$ and $\mathcal{U}_k$, $k = 0, \dots, N-1$, compute zonotopes $\mathcal{R}^{\mathrm{CP}}_1, \dots, \mathcal{R}^{\mathrm{CP}}_N$ that form an $(\alpha, \delta)$ probabilistic reachable sequence for $\mathcal{P}$ and have small volume.
\end{problem}

\section{Conformalized Data-Driven Reachability}\label{sec:method}

\subsection{Calibrated Reachable Sets}\label{sec:calibrated}

The method separates what a model can predict from what it cannot.
The prediction from the initial state and the inputs is propagated by zonotope operations, and only the deviation from the prediction is calibrated on data, which requires no bound on the noise.
Each state of a trajectory decomposes as $x(k) = \hat x(k) + e(k)$ into the prediction $\hat x(k)$ of a model, computed from the initial state and the applied inputs, and the prediction error $e(k)$.
The reachable set at time $k$ is the Minkowski sum of a nominal set, which encloses the predictions, and an error zonotope, which encloses the prediction errors and whose size is calibrated on $\mathcal{D}_{\mathrm{cal}}$.

The model is estimated by least squares from the training data,
\begin{equation}\label{eq:ls}
\tilde M = \begin{bmatrix} \tilde A & \tilde B \end{bmatrix} = X_+ \begin{bmatrix} X_- \\ U_- \end{bmatrix}^\dagger ,
\end{equation}
where $\tilde A$ and $\tilde B$ are the estimates of $A_{\mathrm{tr}}$ and $B_{\mathrm{tr}}$.
The stacked data matrix is assumed to have full row rank.
For a trajectory $\tau$, the open-loop prediction $\hat x(k)$ starts from the initial state of $\tau$ and applies the inputs of $\tau$ to the model,
\begin{equation}\label{eq:pred}
\hat x(0) = x(0), \qquad \hat x(k+1) = \tilde A \hat x(k) + \tilde B u(k),
\end{equation}
for $k = 0, \dots, N-1$.
The prediction depends only on $x(0)$ and on the inputs, and not on the noise.

The nominal sets $\mathcal{N}_k$ enclose the predictions of all trajectories that start in $\mathcal{X}_0$ and apply inputs in $\mathcal{U}_k$.
They are obtained by propagating the initial set and the input sets through the model,
\begin{equation}\label{eq:nominal}
\mathcal{N}_{k+1} = \tilde A \mathcal{N}_k + \tilde B \mathcal{U}_k, \qquad k = 0, \dots, N-1,
\end{equation}
where $\mathcal{N}_0 = \mathcal{X}_0$.
Induction over $k$ with \eqref{eq:pred} and \eqref{eq:nominal} gives $\hat x(k) \in \mathcal{N}_k$.

The prediction error $e(k)$ collects the effect of the noise and of the model error over the first $k$ steps.
Its magnitude grows with $k$ and differs between the state components.
We therefore normalize it by the diagonal matrix $\Lambda_k$ of the empirical standard deviations of the components of $e(k)$ over the training trajectories, so that $\Lambda_k^{-1} e(k)$ expresses each component of the error in units of its standard deviation, a per-component form of the normalized scores in \cite{lei2018distribution,cleaveland2024conformal}.
This requires $K \ge 2$ training trajectories of length at least $N$ and nonzero standard deviations.
To cover all time steps with one calibration, each trajectory $\tau$ is summarized by a single score, the largest normalized prediction error along the trajectory,
\begin{equation}\label{eq:score}
s(\tau) = \max_{k = 1, \dots, N} \big\| \Lambda_k^{-1} e(k) \big\|_\infty .
\end{equation}
The score is the largest deviation of the trajectory from its own prediction, measured in standard deviations.

Since the noise may be unbounded, the prediction errors may be unbounded, and bounded sets can only contain the trajectories whose errors are not too large.
The score quantifies the size of the errors, and a threshold $q \ge 0$ on the score specifies which errors are accepted.
A larger threshold accepts more trajectories and requires larger sets.
For each threshold $q$, we construct sets that contain every trajectory whose score is at most $q$.
By the definition of the score, $s(\tau) \le q$ holds if and only if every component of every normalized error is at most $q$ in absolute value.
Since for an invertible matrix $\Lambda$
\begin{equation}\label{eq:sublevel}
\big\{ z \in \mathbb{R}^{n_x} \ \big|\ \|\Lambda^{-1} z\|_\infty \le q \big\} = \langle 0, q\Lambda \rangle ,
\end{equation}
this is the case if and only if $e(k) \in \langle 0, q\Lambda_k \rangle$ for $k = 1, \dots, N$.
Since the predictions lie in the nominal sets, such a trajectory satisfies $x(k) \in \mathcal{R}^{\mathrm{CP}}_k(q)$ with
\begin{equation}\label{eq:family}
\mathcal{R}^{\mathrm{CP}}_k(q) = \mathcal{N}_k + \langle 0, q \Lambda_k \rangle, \qquad k = 1, \dots, N .
\end{equation}
If $q$ is chosen such that a new trajectory satisfies $s(\tau) \le q$ with probability at least $1-\alpha$, the sets contain the whole trajectory with probability at least $1-\alpha$.
Proposition~\ref{prop:tc} provides such a threshold $\hat q$ with confidence $1-\delta$ from the scores of the calibration trajectories, and the calibrated reachable sets are $\mathcal{R}^{\mathrm{CP}}_k = \mathcal{R}^{\mathrm{CP}}_k(\hat q)$.
Algorithm~\ref{alg:cddr} summarizes the computation, which consists of one least-squares fit, $K + n$ open-loop predictions and $N$ zonotope operations.
The number of generators of $\mathcal{R}^{\mathrm{CP}}_k$ grows linearly in $k$.

\begin{algorithm}[t]
\caption{Conformalized data-driven reachability}\label{alg:cddr}
\raggedright
\begin{algorithmic}[1]
\Require training data $\mathcal{D}$, calibration trajectories $\tau^{(1)}, \dots, \tau^{(n)}$ with $n \ge n_{\min}$, initial set $\mathcal{X}_0$, input sets $\mathcal{U}_0, \dots, \mathcal{U}_{N-1}$, levels $\alpha, \delta$
\Ensure reachable sets $\mathcal{R}^{\mathrm{CP}}_1, \dots, \mathcal{R}^{\mathrm{CP}}_N$
\State form $X_+$, $X_-$, $U_-$ from $\mathcal{D}$ by \eqref{eq:data}
\State $[\tilde A\ \tilde B] \gets X_+ \begin{bmatrix} X_- \\ U_- \end{bmatrix}^\dagger$
\For{$i = 1, \dots, K$}
\State $\hat x^{(i)}(0) \gets x^{(i)}(0)$
\For{$k = 0, \dots, N-1$}
\State $\hat x^{(i)}(k+1) \gets \tilde A \hat x^{(i)}(k) + \tilde B u^{(i)}(k)$
\State $e^{(i)}(k+1) \gets x^{(i)}(k+1) - \hat x^{(i)}(k+1)$
\EndFor
\EndFor
\State $\Lambda_k \gets$ diagonal matrix of the standard deviations of the components of $\{e^{(i)}(k)\}_{i=1}^{K}$, $k = 1, \dots, N$
\For{$j = 1, \dots, n$}
\State $x(\cdot), u(\cdot) \gets$ states and inputs of $\tau^{(j)}$, $\hat x(0) \gets x(0)$
\For{$k = 0, \dots, N-1$}
\State $\hat x(k+1) \gets \tilde A \hat x(k) + \tilde B u(k)$
\EndFor
\State $s^{(j)} \gets \max_{k = 1, \dots, N} \|\Lambda_k^{-1} (x(k) - \hat x(k))\|_\infty$
\EndFor
\State $k^\ast \gets \max\{ k \in \{0, \dots, n-1\} \mid B_{n,\alpha}(k) \le \delta \}$
\State $\hat q \gets$ $(k^\ast+1)$-th largest element of $\{s^{(1)}, \dots, s^{(n)}\}$
\State $\mathcal{N}_0 \gets \mathcal{X}_0$
\For{$k = 0, \dots, N-1$}
\State $\mathcal{N}_{k+1} \gets \tilde A \mathcal{N}_k + \tilde B \mathcal{U}_k$
\State $\mathcal{R}^{\mathrm{CP}}_{k+1} \gets \mathcal{N}_{k+1} + \langle 0, \hat q \Lambda_{k+1} \rangle$
\EndFor
\State \Return $\mathcal{R}^{\mathrm{CP}}_1, \dots, \mathcal{R}^{\mathrm{CP}}_N$
\end{algorithmic}
\end{algorithm}

The construction and Proposition~\ref{prop:tc} together give the guarantee of Definition~\ref{def:prs} for the data distribution.

\begin{theorem}\label{thm:main}
Let Assumption~\ref{ass:support} hold and let $n \ge n_{\min}$.
Then the sets $\mathcal{R}^{\mathrm{CP}}_1, \dots, \mathcal{R}^{\mathrm{CP}}_N$ computed by Algorithm~\ref{alg:cddr} form an $(\alpha, \delta)$ probabilistic reachable sequence for $\mathcal{P}$.
\end{theorem}

\begin{proof}
Since $\tilde M$ and $\Lambda_1, \dots, \Lambda_N$ are computed from $\mathcal{D}$ alone, the score is a fixed function of the trajectory conditionally on $\mathcal{D}$, so the scores of the calibration trajectories and of a new trajectory $\tau \sim \mathcal{P}$ are independent and identically distributed.
By Proposition~\ref{prop:tc}, with probability at least $1-\delta$ over $\mathcal{D}_{\mathrm{cal}}$, the event $s(\tau) \le \hat q$ has probability at least $1-\alpha$.
By Assumption~\ref{ass:support}, $x(0) \in \mathcal{X}_0$ and $u(k) \in \mathcal{U}_k$ hold for $\tau$, so $\hat x(k) \in \mathcal{N}_k$.
On the event $s(\tau) \le \hat q$, the construction of the sets then gives $x(k) \in \mathcal{R}^{\mathrm{CP}}_k$ for $k = 1, \dots, N$.
\end{proof}

That is, with confidence $1-\delta$ over the calibration data, a new trajectory from the data distribution lies in the sets at all time steps with probability at least $1-\alpha$.
The required number of calibration trajectories $n_{\min}$ does not depend on the horizon $N$, because the containment at all time steps is ensured by the single event $s(\tau) \le \hat q$, which is calibrated once.
For $\alpha = \delta = 0.05$, $n_{\min} = 59$.

The guarantee of Theorem~\ref{thm:main} refers to the data distribution $\mathcal{P}$, which is determined by the initial distribution and the policy used for data collection.
When the sets are used to verify another controller, the trajectories follow $\mathcal{P}' = \mathcal{P}(P_0', \pi', P_w)$ with another initial distribution $P_0'$ and another causal policy $\pi'$, while the noise distribution $P_w$ is unchanged.
The containment of every trajectory with $s(\tau) \le \hat q$ in the sets uses only that the inputs lie in $\mathcal{U}_k$, so it holds for every input sequence in $\mathcal{U}_k$, including inputs generated by feedback.
What may change under $\mathcal{P}'$ is the probability of the event $s(\tau) \le \hat q$.

\begin{remark}\label{rem:estimated}
Subtracting \eqref{eq:pred} from \eqref{eq:sys} splits the prediction error as $e(k) = e_w(k) + e_M(k)$ with
\begin{align*}
e_w(k) &= \sum_{i=0}^{k-1} A_{\mathrm{tr}}^{k-1-i} w(i), \\
e_M(k) &= \sum_{i=0}^{k-1} A_{\mathrm{tr}}^{k-1-i} \Delta \tilde M \begin{bmatrix} \hat x(i) \\ u(i) \end{bmatrix},
\end{align*}
where $\Delta \tilde M = [A_{\mathrm{tr}}\ B_{\mathrm{tr}}] - \tilde M$ is the model error.
By Assumption~\ref{ass:support}, $e_w$ has the same distribution under $\mathcal{P}$ and $\mathcal{P}'$, so the score depends on the initial distribution and the controller only through $e_M$, which scales with the model error.
If $\rho \ge \|\Lambda_k^{-1} e_M(k)\|_\infty$ for all $k$, all initial states in $\mathcal{X}_0$ and all inputs in $\mathcal{U}_k$, the sets $\mathcal{R}^{\mathrm{CP}}_k(\hat q + 2\rho)$ form an $(\alpha, \delta)$ probabilistic reachable sequence for every $\mathcal{P}'$.
This follows from Proposition~\ref{prop:tc} applied to the scores of $e_w$, which differ from the scores of $e$ by at most $\rho$, so that their threshold is at most $\hat q+\rho$ and the score of a trajectory from $\mathcal{P}'$ is at most $\hat q+2\rho$ whenever the score of its $e_w$ is at most that threshold.
\end{remark}

\subsection{One-Step Residuals and Direct Set Calibration}\label{sec:direct}

Two other scores show the trade-off between the size of the sets and the dependence on the controller.

The one-step residual $r(k) = x(k+1) - \tilde A x(k) - \tilde B u(k)$ also depends on the controller only through the model error but yields larger sets.
The sets of the one-step residual add at each step an error zonotope calibrated on the normalized residuals, so the sets at step $k$ contain $k$ error zonotopes, one for the noise of each step, whereas the sets \eqref{eq:family} add a single error zonotope calibrated on the accumulated error $e(k)$.

Calibrating the sets directly yields smaller sets.
Direct calibration uses as score the smallest threshold for which a trajectory lies in the sets,
\begin{equation}\label{eq:snested}
s_{\mathcal{N}}(\tau) = \min\big\{ q \ge 0 \ \big|\ x(k) \in \mathcal{R}^{\mathrm{CP}}_k(q),\ k = 1, \dots, N \big\} ,
\end{equation}
which is the largest normalized distance of the states to the nominal sets and is computed by one linear program per time step.
Since the prediction $\hat x(k)$ lies in $\mathcal{N}_k$, this distance is at most the distance to the prediction, so $s_{\mathcal{N}} \le s$ and the calibrated threshold is at most that of CDDR.
By the proof of Theorem~\ref{thm:main}, the sets form an $(\alpha, \delta)$ probabilistic reachable sequence for $\mathcal{P}$.
The score $s_{\mathcal{N}}$, however, depends on the position of $x(k)$ within $\mathcal{N}_k$ and thereby on the initial distribution and the controller, also when the model error is zero.

Among the three scores, only the prediction error both depends on the controller through the model error alone and adds a single error zonotope at each step.

\subsection{Measurement Noise}\label{sec:meas}

When only the outputs $y(k)$ are measured, a trajectory in the data consists of the outputs and the inputs, while the coverage still refers to the state $x(k)$.
Without information on $v(k)$, each output $y(k)$ is consistent with any state $x(k)$, so the method requires information on $v$, and we assume a known quantile.

\begin{assumption}\label{ass:quantile}
For a known $\varepsilon > 0$ and $\alpha_v \in [0,1)$, the measurement noise satisfies $\mathbb{P}\big[\|v(k)\|_\infty > \varepsilon\big] \le \alpha_v$ for $k = 0, \dots, N$.
The measurement noise sequences of different trajectories are independent and identically distributed.
\end{assumption}

Assumption~\ref{ass:quantile} restricts only the marginal tails of $v$, and a hard bound corresponds to $\alpha_v = 0$.

Since the states are not measured, the method is applied in output space.
The model \eqref{eq:ls}, the prediction \eqref{eq:pred} and the score \eqref{eq:score} are computed with $y$ in place of $x$, which gives the output prediction $\hat y(k)$, the output score $s_y(\tau)$ and the matrices $\Lambda_k$ of the standard deviations of $y(k) - \hat y(k)$ over the training trajectories.
The nominal output sets $\mathcal{N}^y_k$ follow \eqref{eq:nominal} from the initial set enlarged by the quantile, $\mathcal{N}^y_0 = \mathcal{X}_0 + \langle 0, \varepsilon I_{n_x} \rangle$, since $\hat y(0) = y(0) = x(0) + v(0)$ lies in this set whenever $\|v(0)\|_\infty \le \varepsilon$.
For $q \ge 0$, the output sets $\mathcal{Y}_k(q)$ and the state sets $\mathcal{R}^{\mathrm{CP}}_k(q)$ are
\begin{equation}\label{eq:sets_y}
\mathcal{Y}_k(q) = \mathcal{N}^y_k + \langle 0, q\Lambda_k \rangle, \qquad
\mathcal{R}^{\mathrm{CP}}_k(q) = \mathcal{Y}_k(q) + \langle 0, \varepsilon I_{n_x} \rangle ,
\end{equation}
where the enlargement of the state sets accounts for $x(k) = y(k) - v(k)$.
The threshold $\hat q$ follows from Proposition~\ref{prop:tc} applied to the output scores of $\mathcal{D}_{\mathrm{cal}}$.
The sets contain the state trajectory whenever the output score does not exceed $\hat q$ and the measurement noise does not exceed its quantile at any time step.

\begin{theorem}\label{thm:meas}
Let Assumptions~\ref{ass:support} and~\ref{ass:quantile} hold and let $n \ge n_{\min}$.
Conditionally on $\mathcal{D}$, with probability at least $1-\delta$ over $\mathcal{D}_{\mathrm{cal}}$,
\begin{equation*}
\mathbb{P}_{\tau \sim \mathcal{P}}\big[ x(k) \in \mathcal{R}^{\mathrm{CP}}_k(\hat q),\ k = 1, \dots, N \big] \ge 1 - \alpha - (N+1)\alpha_v .
\end{equation*}
\end{theorem}

\begin{proof}
Conditionally on $\mathcal{D}$, the output scores of $\mathcal{D}_{\mathrm{cal}}$ and of the new trajectory are independent and identically distributed.
By Proposition~\ref{prop:tc}, with probability at least $1-\delta$ over $\mathcal{D}_{\mathrm{cal}}$, the event $E_1$ that $s_y(\tau) \le \hat q$ has probability at least $1-\alpha$.
By Assumption~\ref{ass:quantile} and the union bound over the $N+1$ time steps, the event $E_2$ that $\|v(k)\|_\infty \le \varepsilon$ for $k = 0, \dots, N$ has probability at least $1 - (N+1)\alpha_v$.
By the union bound, $E_1 \cap E_2$ has probability at least $1 - \alpha - (N+1)\alpha_v$.
On $E_1 \cap E_2$, $\hat y(0) \in \mathcal{N}^y_0$, so the argument of Theorem~\ref{thm:main} with $\hat y$ and $\mathcal{N}^y_k$ in place of $\hat x$ and $\mathcal{N}_k$ gives $y(k) \in \mathcal{Y}_k(\hat q)$ and hence $x(k) = y(k) - v(k) \in \mathcal{R}^{\mathrm{CP}}_k(\hat q)$ for all $k$.
\end{proof}

Compared with Theorem~\ref{thm:main}, the sets are enlarged by $\varepsilon$ twice, in the initial set and in the state sets.

\subsection{Nonlinear Systems}\label{sec:nonlinear}

For the nonlinear system \eqref{eq:sys_gen} without measurement noise, a global linear model has a large prediction error when $f$ is strongly nonlinear, which enlarges the sets.
The model \eqref{eq:ls} is therefore replaced by local affine models fitted separately for each time step, which describe the dynamics near the states that the trajectories reach at that step.
For each step $k$, a local affine model $x(k+1) \approx \tilde A_k x(k) + \tilde B_k u(k) + \tilde c_k$ is fitted by least squares to the training states and inputs at step $k$ and the training states at step $k+1$.
The fit requires the stacked samples at each step to have full row rank.
The prediction \eqref{eq:pred} and the nominal sets \eqref{eq:nominal} then apply at each step the model of that step, with the offset $\tilde c_k$ added.
With one model per time step, each nominal set is the affine image of a zonotope, whereas models defined on regions of the state space would require splitting the nominal sets at the region boundaries.
The score \eqref{eq:score}, the sets \eqref{eq:family} and Algorithm~\ref{alg:cddr} are otherwise unchanged.

The proof of Theorem~\ref{thm:main} uses only the definitions of the prediction, the nominal sets and the prediction error, so Theorem~\ref{thm:main} applies to nonlinear systems without change.

\section{Numerical Experiments}\label{sec:experiments}

\subsection{Setup}\label{sec:setup}

We evaluate the method on a five-dimensional linear benchmark system \cite{alanwar2023datadriven} and on a nonlinear system that is not Lipschitz continuous.
The linear system is discretized with sampling time $0.05$\,s, with the initial set $\mathcal{X}_0 = \langle \mathbf{1}, 0.1 I_5 \rangle$, the input sets $\mathcal{U}_k = \langle 10, 0.25 \rangle$ and the horizon $N = 5$, where $\mathbf{1}$ is the vector of ones.
The process noise is Gaussian with standard deviation $\sigma_w = 0.01$ in every component, so its support is unbounded.
The comparison with matrix-zonotope reachability requires bounded noise and uses two bounded distributions.
The segment noise has five equal components that are uniform on $[-0.005, 0.005]$, and its bound is the noise zonotope used in the evaluation of matrix-zonotope reachability \cite{alanwar2023datadriven}.
The box noise has independent components that are uniform on $[-0.0173, 0.0173]$, with the standard deviation $0.01$ of the Gaussian noise.
The data are collected with initial states drawn uniformly from $\mathcal{X}_0$ and inputs drawn uniformly and independently from $\mathcal{U}_k$ at every time step.
The training data and the calibration data consist of $\NUM{100}$ trajectories of length $N$ each, and the results are evaluated on $\NUM{10000}$ test trajectories unless stated otherwise.
With $\alpha = \delta = 0.05$, Proposition~\ref{prop:tc} gives $k^\ast = \NUM{1}$, so the threshold is the second largest calibration score.

The guarantee concerns the score, whereas a user of the sets checks whether the states lie in them, so we report two coverages.
The set coverage is the fraction of test trajectories that lie in the sets at all time steps, which is the quantity in Definition~\ref{def:prs}.
The score coverage is the fraction of test trajectories with $s(\tau) \le \hat q$, which is the event whose probability Proposition~\ref{prop:tc} controls.
Since this event implies containment but not conversely, the set coverage is at least the score coverage.

We compare CDDR, which denotes Algorithm~\ref{alg:cddr}, with four methods.
Matrix-zonotope reachability (MZ) receives the true noise bound and uses no calibration.
Marginal conformal prediction (MCP) uses the score and the sets of Algorithm~\ref{alg:cddr} with the threshold $\hat q_{\mathrm{m}}$ of \eqref{eq:marginal}, which guarantees the coverage $1-\alpha$ only on average over calibration data.
Two methods calibrate the states directly, whereas CDDR calibrates the prediction error.
The model-based direct calibration (DC-MB) calibrates the score $s_{\mathcal{N}}$, the distance of the states to the nominal sets.
The model-free direct calibration (DC-MF) uses an axis-aligned box around the mean of the training states at each time step, with half-widths proportional to their standard deviations, and calibrates its scaling on the calibration states.

\subsection{Size of the Sets}\label{sec:size}

Without any information on the noise, the sets of CDDR are smaller than those of MZ, which receives the true noise bound.
This comparison uses the segment noise and the box noise, and MZ receives their true bounds.
Both methods use the same $\NUM{200}$ trajectories of length $5$.
MZ uses all of them, and CDDR uses $\NUM{100}$ for training and $\NUM{100}$ for calibration.
We repeat the comparison with $\NUM{20}$ independent data sets.

For the segment noise, the sets of MZ are on average $\NUM{2.6}$ times and at least $\NUM{2.3}$ times larger than those of CDDR, and both methods contain all $\NUM{10000}$ test trajectories, as Fig.~\ref{fig:alanwar} shows for one data set.
The sets of MZ contain the exact reachable set, which covers every noise sequence within the bound, whereas the sets of CDDR are required to cover a fraction $1-\alpha$ of the trajectories, which allows smaller sets.

For the box noise, the sets of MZ are on average $\NUM{4.7 \cdot 10^7}$ times and at least $\NUM{2.3 \cdot 10^7}$ times larger than those of CDDR, because the sets of MZ grow along the horizon.
The set of models that are consistent with the data multiplies the current set at each step, so the uncertainty added at each step grows with the set.
From the first to the fifth time step, the volume of the sets of MZ grows by a factor of $\NUM{3.1 \cdot 10^6}$, while that of the sets of CDDR grows by a factor of $\NUM{2.2}$.
More data do not close the gap, since increasing the data of MZ from $\NUM{30}$ to $\NUM{200}$ trajectories reduces the final-step volume of its sets by only $\NUM{1.3}$~\% for the segment and $\NUM{16}$~\% for the box.

The sets of CDDR are small because of two choices in Algorithm~\ref{alg:cddr}.
Replacing the open-loop prediction error by the one-step residual enlarges the final-step volume by a factor of $\NUM{11}$.
Omitting the normalization $\Lambda_k$ enlarges it by a factor of $\NUM{554}$ when the noise standard deviation is $0.10$ in one component and $0.005$ in the others.
Both factors are medians over $\NUM{1000}$ calibration draws.

\begin{figure}[t]
\centering
\includegraphics[width=\columnwidth]{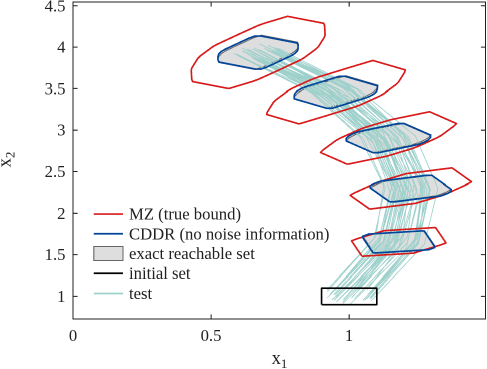}
\caption{Projection onto $x_1$ and $x_2$ of the sets of MZ with the true noise bound and of CDDR without noise information, for the segment noise and the data set with the median volume ratio among $20$ data sets of $200$ trajectories, with $40$ test trajectories. The exact reachable set is computed with the true model and the true bound.}
\label{fig:alanwar}
\end{figure}

\subsection{Validity of the Guarantee}\label{sec:validity}

The score coverage of CDDR behaves as Proposition~\ref{prop:tc} predicts.
We repeat the calibration $\NUM{1000}$ times with new calibration data and evaluate every threshold on $\NUM{10^6}$ common test trajectories.
Fig.~\ref{fig:coverage} shows the resulting score coverages of CDDR and of MCP, together with the Beta densities that the coverage of an order statistic follows for continuous scores \cite{angelopoulos2025foundations}.
The score coverage of CDDR falls below $1-\alpha$ in $\NUM{3.7}$~\% of the calibration draws, in agreement with the bound $B_{n,\alpha}(k^\ast) = \NUM{3.7}$~\%, whereas that of MCP falls below $1-\alpha$ in $\NUM{44.2}$~\%.
The same $n = 100$ suffices for longer horizons, since for $N = 50$ the score coverage of CDDR falls below $1-\alpha$ in $\NUM{3.5}$~\% of $\NUM{1000}$ calibration draws.
At this horizon, the final sets have $\NUM{60}$ generators, and fitting, calibration and propagation take $\NUM{0.05}$\,s in total.

The guarantee for the realized threshold requires a larger threshold than MCP, and the resulting increase in set volume decreases with the number of calibration trajectories.
For the guarantee to hold for the threshold computed from one calibration data set, the threshold must exceed the $(1-\alpha)$ quantile of the score with confidence $1-\delta$, so it is a higher order statistic than the threshold of MCP.
For $n = 100$, CDDR uses the second largest calibration score and MCP the fifth largest.
The final-step volume of CDDR exceeds that of MCP by $\NUM{39}$~\% for $n = 59$, by $\NUM{27}$~\% for $n = 100$, by $\NUM{6}$~\% for $n = 1000$ and by $\NUM{3}$~\% for $n = 5000$.

\begin{figure}[t]
\centering
\includegraphics[width=\columnwidth]{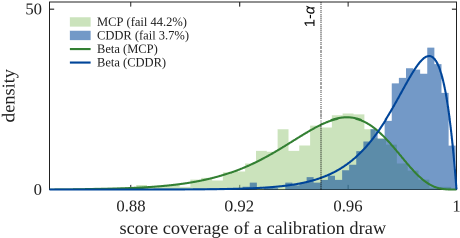}
\caption{Score coverage of CDDR and MCP for $1000$ calibration draws with $n = 100$, with the Beta densities predicted by theory. The dotted line marks $1-\alpha$.}
\label{fig:coverage}
\end{figure}

\subsection{Other Controllers}\label{sec:shift}

The coverage of the sets of CDDR stays above the target $1-\alpha$ when the initial states and the controller change, while the coverage of DC-MB and DC-MF falls below it.
The calibration data follow the data distribution, denoted S0, and the test trajectories follow S0 or one of the following distributions.
S1 changes the initial states and draws them from the $\NUM{32}$ vertices of $\mathcal{X}_0$.
S2 changes the controller and generates the inputs by the state feedback $u(k) = 10 - \mathbf{1}^\top (x(k) - \mathbf{1})$, saturated to $\mathcal{U}_k$.
S3 changes both.
It starts at a vertex of $\mathcal{X}_0$ and applies a constant input at an endpoint of $\mathcal{U}_k$, chosen such that the noise-free trajectory ends farthest, in the infinity norm, from the noise-free trajectory that starts at the center of $\mathcal{X}_0$ with the center of $\mathcal{U}_k$ as input.
Only the noise is random in S3.
Table~\ref{tab:shift} reports the coverages averaged over $\NUM{20}$ calibration draws.

Under S1, S2 and S3, the set coverage and the score coverage of CDDR stay above the target $1-\alpha = 95$~\%, as Table~\ref{tab:shift} shows.
The coverage of DC-MB falls below the target under S1 and S3, and that of DC-MF under all of them.
Under the feedback controller of S2, DC-MB keeps a coverage of $\NUM{97.0}$~\%, so whether its coverage falls below the target depends on the controller.
The coverage drops most under S3, where DC-MB covers $\NUM{66.7}$~\% and DC-MF $\NUM{11.3}$~\%.
In S3, as Fig.~\ref{fig:shift} shows, the trajectories start at a vertex of $\mathcal{X}_0$ and stay near the boundary of the nominal sets, which the calibration trajectories of DC-MB rarely reach.
In Fig.~\ref{fig:shift}, $\NUM{26}$ of $\NUM{40}$ test trajectories leave the sets of DC-MB, while none leaves the sets of CDDR.
Remark~\ref{rem:estimated} explains this difference, since the score of CDDR depends on the controller only through the model error, whereas the scores of DC-MB and DC-MF depend on the position of the states.
The final-step volumes in Table~\ref{tab:shift} show that DC-MB and DC-MF obtain smaller sets than CDDR at the cost of this loss of coverage.
The same results hold for heavy-tailed Student-t noise, where the set coverage of CDDR is at least $\NUM{99.73}$~\% under all four distributions, while DC-MB and DC-MF lose coverage as before.

\begin{table}[t]
\centering
\caption{Set coverage and score coverage of CDDR and set coverage of DC-MB and DC-MF in \%, averaged over $20$ calibration draws, under the test distributions S0 to S3, and final-step volume.}
\label{tab:shift}
\footnotesize\setlength{\tabcolsep}{3pt}
\begin{tabular}{lcccc}
\toprule
 & \multicolumn{2}{c}{CDDR} & DC-MB & DC-MF \\
 & set & score & set & set \\
\midrule
S0 & 100.00 & 98.0 & 98.2 & 97.7 \\
S1 & 99.98 & 97.7 & 62.1 & 73.8 \\
S2 & 100.00 & 98.1 & 97.0 & 83.0 \\
S3 & 100.00 & 97.9 & 66.7 & 11.3 \\
\midrule
volume & \multicolumn{2}{c}{$3.6{\times}10^{-3}$} & $6.2{\times}10^{-4}$ & $5.9{\times}10^{-4}$ \\
\bottomrule
\end{tabular}

\end{table}

\begin{figure}[t]
\centering
\includegraphics[width=\columnwidth]{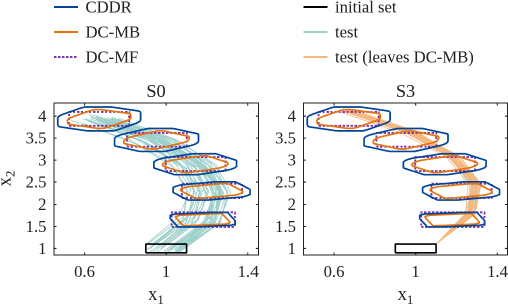}
\caption{Projection onto $x_1$ and $x_2$ of the sets of CDDR, DC-MB and DC-MF under S0 and S3, for the calibration draw with the median threshold of CDDR, with $40$ test trajectories. Light orange marks trajectories that leave the sets of DC-MB at some time step, determined in all five dimensions.}
\label{fig:shift}
\end{figure}

\subsection{Measurement Noise}\label{sec:exp_meas}

The sets of CDDR for measurement noise attain the coverage of Theorem~\ref{thm:meas}.
We add measurement noise $v(k)$ with independent components of equal variance, which are uniform on $[-0.002, 0.002]$, Gaussian, or Student-t with five degrees of freedom.
MZ is not included, since its guarantee under measurement noise requires a bound that depends on the unknown system matrix \cite{alanwar2023datadriven}.
Theorem~\ref{thm:meas} splits the admissible miscoverage between the calibration and the measurement noise, and we take $\alpha = \NUM{0.04}$ and $\alpha_v = \NUM{0.01}/(N+1)$ to obtain the coverage $0.95$ of the other experiments.
For each distribution, $\varepsilon$ is a value that the measurement noise exceeds at one time step with probability at most $\alpha_v$, which gives $\varepsilon = \NUM{0.002}$ for the uniform noise, $\varepsilon = \NUM{0.004143}$ for the Gaussian noise and $\varepsilon = \NUM{0.007774}$ for the heavier-tailed Student-t noise.

With $n = 100$ and $\alpha = 0.04$, Proposition~\ref{prop:tc} gives $k^\ast = \NUM{0}$, so the threshold is the largest calibration score.
Over $\NUM{20}$ calibration draws and all three distributions, all test states lie in the sets, so the coverage guaranteed by Theorem~\ref{thm:meas} is attained.

\subsection{Nonlinear System}\label{sec:exp_nonlinear}

CDDR applies to nonlinear systems without a Lipschitz constant, and local affine models yield smaller sets than a global linear model.
We use the autonomous system $x(k+1) = A x(k) + \gamma \varphi(x(k)) + w(k)$ with $A = \begin{bmatrix} 0.7 & 0.35 \\ -0.35 & 0.7 \end{bmatrix}$ and $\varphi(x) = [\mathrm{sign}(x_1)\sqrt{|x_1|}\ \ \mathrm{sign}(x_2)\sqrt{|x_2|}]^\top$.
We take $\gamma = 0.2$, Gaussian noise with $\sigma_w = 0.01$, the initial set $\mathcal{X}_0 = \langle [0.3\ 1.2]^\top, 0.1 I_2 \rangle$ and the horizon $N = 5$.
The square root is not Lipschitz continuous at zero, and every trajectory crosses $x_2 = 0$ within the horizon, so Lipschitz-based matrix-zonotope reachability does not apply.
Averaged over $\NUM{20}$ calibration draws, the set coverage of CDDR with local affine models is $\NUM{99.96}$~\% and its score coverage $\NUM{97.4}$~\%, above the target $1-\alpha$ of Theorem~\ref{thm:main}.
The final-step volume of a global linear model is $\NUM{2.8}$ times that of the local models, as Fig.~\ref{fig:nonlinear} shows, and the ratio grows with the nonlinearity from $\NUM{1.1}$ for $\gamma = 0.05$ to $\NUM{67}$ for $\gamma = 1$.

\begin{figure}[t]
\centering
\includegraphics[width=\columnwidth]{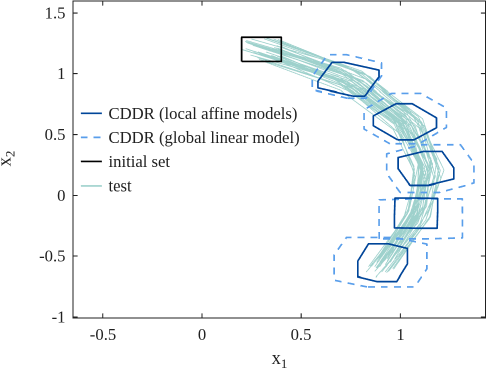}
\caption{Sets of CDDR with local affine models and with a global linear model for the nonlinear system with $\gamma = 0.2$, with $40$ test trajectories.}
\label{fig:nonlinear}
\end{figure}

\section{Conclusion}\label{sec:conclusion}

This paper proposed Conformalized Data-Driven Reachability (CDDR), which calibrates the open-loop prediction error of a least-squares model with training-conditional split conformal prediction and propagates the initial set and the input set with zonotopes.
The sets contain the whole trajectory with probability at least $1-\alpha$ and confidence $1-\delta$ without a bound on the process noise.
The prediction error depends on the initial states and the controller only through the model error.
In the experiments, the sets kept their coverage when the initial states and the controller changed, whereas sets calibrated directly on the states lost coverage.
CDDR also yielded smaller sets than matrix-zonotope reachability with the true noise bound.
Future work includes a data-driven bound on the model error of Remark~\ref{rem:estimated}, a less conservative propagation of the initial set and the input set, and local models on regions of the state space.

\bibliographystyle{IEEEtran}
\bibliography{references}

\end{document}